\documentclass[10pt,twocolumn,twoside]{IEEEtran}
\usepackage[tight,footnotesize]{subfigure}
\usepackage[dvips]{graphicx}
\usepackage{amsmath,amsfonts}
\usepackage{multirow}
\usepackage{latexsym}
\usepackage[T1]{fontenc}
\usepackage{textcomp}
\usepackage{calc}
\usepackage{subfig}
\usepackage{amsfonts,amssymb,amsbsy,bm,paralist,theorem,cite,ifthen,color}
\usepackage{threeparttable}
\usepackage{color}
\usepackage{stfloats}

\allowdisplaybreaks[4]  

\usepackage{soul}
\setstcolor{red}

\newtheorem{proposition}{Proposition}

\theorembodyfont{\rmfamily}

\IEEEoverridecommandlockouts    
\begin{document}

\title{Power Allocation and Time-Domain Artificial Noise Design for Wiretap OFDM with {Discrete} Inputs}
\author{Haohao~Qin$^\S$,
        Yin~Sun$^\dag$,~\IEEEmembership{Member,~IEEE,}
        Tsung-Hui~Chang$^\ddag$,~\IEEEmembership{Member,~IEEE,}
        Xiang~Chen$^\S$,~\IEEEmembership{Member,~IEEE,}        
        Chong-Yung~Chi$^*$,~\IEEEmembership{Senior Member,~IEEE,}
        Ming~Zhao$^\S$,~\IEEEmembership{Member,~IEEE,}
        and Jing~Wang$^\S$,~\IEEEmembership{Member,~IEEE}
\thanks{Manuscript received May 12, 2012; revised August 18, 2012 and November 26, 2012; accepted March 7, 2013. The associate editor coordinating the review of this paper and approving it for publication was Sofi\`ene Affes.}
\thanks{This work is supported in part by National Basic Research Program of China (2012CB316002),
National S\&T Major Project (2011ZX03004-004), National Natural Science Foundation of China(61132002), Tsinghua Research
Funding-No.2010THZ02-3, National Science
Council, Taiwan, under Grant NSC99-2221-E007-052-MY3, and National Science Council, Taiwan, under Grant NSC101-2218-E-011-043, International S\&T Cooperation Program (2012DFG12010) and Ericsson Company.}
\thanks{$^\S$Haohao Qin, Xiang Chen (the corresponding author), Ming Zhao and Jing Wang are with State Key Laboratory on Microwave
and Digital Communications, Tsinghua National Laboratory for
Information Science and Technology, Department of Electronic
Engineering, Tsinghua University, BJ 100084, P.~R.~China (e-mail:
$\{$haohaoqin07, chenxiang98$\}$@gmail.com). Xiang Chen is also with
Aerospace center, Tsinghua. $^\dag$Yin Sun is with Department of
Electrical and Computer Engineering, the Ohio State University,
Columbus, OH 43210, USA (e-mail: sunyin02@gmail.com).
$^\ddag$Tsung-Hui~Chang is with Department of Electronic and
Computer Engineering, National Taiwan University of Science and
Technology, Taipei 106, Taiwan (E-mail: tsunghui.chang@ieee.org).
$^*$Chong-Yung~Chi is with the Institute of Communications
Engineering and the Department of Electrical Engineering, National
Tsinghua University, Hsinchu,
Taiwan 30013 (e-mail: cychi@ee.nthu.edu.tw).}
}

\markboth{ACCEPTED BY IEEE TRANSACTIONS ON WIRELESS COMMUNICATIONS, Mar. 2013}%
{ACCEPTED BY IEEE TRANSACTIONS ON WIRELESS COMMUNICATIONS, Mar. 2013}

\maketitle
{\begin{abstract}
Optimal power allocation for orthogonal frequency division multiplexing (OFDM) wiretap channels with Gaussian channel inputs has already been studied in some {previous} works from an information theoretical viewpoint. However, these results are not sufficient for practical system designs. One reason is that discrete channel inputs, such as quadrature amplitude modulation (QAM) signals, instead of Gaussian channel inputs, are {deployed in current practical wireless systems} to maintain moderate peak transmission power and receiver complexity. In this paper, we investigate the power allocation and artificial noise design for OFDM wiretap channels with discrete channel inputs. We first prove that the secrecy rate function for discrete channel inputs is nonconcave with respect to the transmission power. To resolve the corresponding nonconvex secrecy rate maximization problem, we develop a low-complexity power allocation algorithm, which yields a duality gap diminishing in the order of $O(1/\sqrt{N})$, where $N$ is the number of subcarriers of OFDM. We then show that independent frequency-domain artificial noise cannot improve the secrecy rate of single-antenna wiretap channels. Towards this end, we propose a novel \emph{time-domain artificial noise} design which exploits temporal degrees of freedom provided by the cyclic prefix of OFDM systems {to jam the eavesdropper and boost the secrecy rate even with a single antenna at the transmitter}. Numerical results are provided to illustrate the performance of the proposed design schemes.
\end{abstract}}

\section{Introduction}
{Security has become increasingly important for wireless networks
due to the proliferation of privacy-sensitive wireless services.
Traditionally, wireless information security is handled by
cryptographic protocols in media access control (MAC) and higher
layers \cite{Massey1988}. However, these techniques face severe
challenges due to the rapid developments of encryption breaking
algorithms and super-computers \cite{Nichols2004}. Recently, various
physical-layer techniques have been proposed to realize perfect
secrecy in wireless networks \cite{Liang2009,LiuR1010}.

The fundamentals of physical-layer security techniques were
laid in \cite{Wyner1975,Cheong1978,Csiszar1978}.
These works studied the maximum data rate for secrecy communications, i.e., the secrecy capacity, for a wiretap channel in which an eavesdropper (Eve) intends to wiretap the secrecy
communications from a transmitter (Alice) to a legitimate receiver
(Bob). Recently, various techniques, such as power allocation, beamforming and training schemes, have been developed to maximize the secrecy capacity
of wiretap channels, e.g., \cite{Li2006,Liang2008,Jorswieck2008,Li2007,Liu2009,Khisti2010,Khisti2010Nov,Goel2008,Liao2011,Chang2010,LiMa2011SP,LiMa2011,He-Yener2012}.
One effective technique is initiatively transmitting artificial noise to jam the eavesdropper if the transmitter is equipped with multiple antennas, e.g., \cite{Goel2008,Liao2011,Chang2010,LiMa2011SP,LiMa2011}. Besides, advanced techniques based on
secrecy-key agreement were also studied in
{\cite{Bloch2008,Lai2008,Wallace2010,Bloch-Barros2011,Lai2012}} to
improve the security of wireless networks.

A common assumption of these studies is that the transmitted signal
{has} a Gaussian distribution. However, Gaussian signals are hardly
used in practice due to its infinite peak power and its excessive
detection complexities. Instead, discrete inputs such as Phase Shift
Keying (PSK) and Quadrature Amplitude Modulation (QAM) (see Fig.
\ref{fig:SystemModel}(a)) are prevalent in practical digital
communication systems \cite{Wilson1996}. Furthermore, most existing
artificial noise designs rely heavily on the {spatial} degrees of freedom
provided by multiple transmit antennas, which are not available in
single-antenna wiretap channels.


In this paper, we consider an OFDM wiretap channel with discrete channel inputs, where each node is equipped with a single antenna {and Alice has perfect knowledge of the channel state information (CSI) for the wireless links to Bob and Eve. We intend to answer the following questions: %
What are the differences between the secrecy rate functions corresponding to Gaussian and discrete channel inputs? Do these differences introduce additional difficulty {in solving} the power allocation problem of the OFDM wiretap channel with discrete channel inputs? Can we make use of artificial noise to improve the secrecy rate of an OFDM wiretap channel when the transmitter is equipped with only one antenna? To address these questions, we first study the convexity of the secrecy rate function with discrete channel inputs, and then develop power allocation and time-domain artificial noise designs to maximize the secrecy rate of OFDM wiretap channels. The main contributions of this paper are summarized as follows:}
\begin{itemize}
\item
We prove that the secrecy rate of any discrete channel
inputs {with a finite number of possible values (or more generally with a finite entropy) is a nonconcave function with respect to the
transmit power (Proposition \ref{proposition:secrecy rate})}. This is in contrast to the case of Gaussian
channel inputs, where the associated secrecy rate is concave and the
optimal power allocation has a closed-form solution
\cite{Li2006,Liang2008,Jorswieck2008}.
\item
A {low-complexity} power allocation algorithm based on Lagrange dual
optimization is {then} proposed for discrete channel inputs. We show that the duality gap of the proposed algorithm
diminishes asymptotically in the order of $O(1/\sqrt{N})$ {as $N$ increases, where
$N$ is the number of subcarriers of OFDM (Proposition \ref{proposition:approximate optimal}).}

\item
{We show {that simply inserting independent artificial noise in the frequency domain} cannot improve the secrecy rate of single-antenna wiretap channels (Proposition \ref{property:frequency AN}).} To resolve this problem, we propose a \emph{time-domain artificial noise} design for the
considered single-antenna OFDM wiretap channel which exploits
\emph{temporal degrees of freedom provided by the cyclic prefix of
OFDM systems} to jam the eavesdropper. {{To the best of our knowledge}, this is the first time-domain artificial noise design for OFDM wiretap channels.}
\item
{Finally, we jointly optimize the subcarrier power allocation and the covariance matrix of the time-domain artificial noise to improve the
secrecy rate.} Successive convex approximation methods are
proposed to handle the joint design problem efficiently. 
Numerical results are presented to show that the proposed
artificial noise schemes can considerably boost the secrecy rate.
\end{itemize}

{{There are several related works published recently. For example, linear} precoding was studied for
multiple-input multiple-output (MIMO) wiretap channels with
discrete inputs \cite{Wu-Xiao2012,Bashar2011arxiv}, where the solution is  locally optimal.} {In \cite{Renna2012}, the OFDM wiretap channel was treated as a special instance of the MIMO wiretap channel and its achievable secrecy rates were studied under both Gaussian inputs and rectangular QAM constellations through asymptotic high/low SNR analysis and numerical evaluations. In contrast, we consider a broader discrete channel input setting along with more general analytical results in Propositions 1-3.} {Artificial noise design was studied in \cite{Li-Yates-Trappe2010} for single antenna wiretap channel with discrete inputs (without OFDM), by assuming an AWGN channel to the Bob and a fast fading channel to the Eve. Our work considers quasi-static fading channels to both Bob and Eve, and  respectively studies  frequency domain and time domain artificial noise designs.}

The remainder of this paper is organized as follows. In Section
\ref{sec:system model and secrecy rate}, we present the system model
and the power allocation problem, and then prove the nonconcavity of
the secrecy rate under discrete channel inputs. In Section
\ref{sec:power allocation}, a power allocation algorithm without
using artificial noise is presented for handling the secrecy rate
maximization problem. Section \ref{sec:AN} presents the time-domain
artificial noise design and two artificial noise aided power
allocation algorithms (one for discrete inputs and the other for
Gaussian inputs). Finally, some conclusions are drawn in Section
\ref{sec:conclu}.

{\bf Notation:} $\mathbb{C},\mathbb{C}^n$ and $\mathbb{C}^{m\times
n}$ denote the set of complex numbers, $n$-vectors, $m\times n$
matrices, respectively. Bold uppercase letters denote matrices and
bold lowercase letters denote column vectors. {$\mathbf{I}_N$
denotes an $N\times N$ identity matrix.}
 $\mathbf{A}\succeq \mathbf{0}$ denotes that the matrix
$\mathbf{A}$ is a positive semi-definite matrix.
$\mathrm{tr}(\mathbf{A})$ denotes the trace of matrix $\mathbf{A}$.
$\mathbf{p} \succeq \mathbf{0}$ means that each component of vector
$\mathbf{p}$ is nonnegative.
$\mathbf{x}\sim\mathcal{CN}(\mathbf{0},\mathbf{\Sigma})$ denotes
that $\mathbf{x}$ is a complex Gaussian random vector with zero mean
and covariance matrix $\mathbf{\Sigma}$. $\mathbb{E}[x]$ represents
the expectation of the random variable $x$, and $\mathbb{E}[x|y]$
denotes the conditional expectation of $x$ given $y$. Function
$\mathcal{H}(x)$ denotes the entropy of random variable $x$, and
$\mathcal{I}(x;y)$ denotes the mutual information between random
variables $x$ and $y$. $\mathrm{Diag}(x_1,x_2,...,x_N)$ denotes the
$N\times N$ diagonal matrix whose diagonal elements are $x_1,
x_2,...,x_N$. $[x]^+ \triangleq \max\{x,0\}$, and $f'(x_0)$ denotes
the first derivative of $f(x)$ at the point $x_0$.

\vspace{-0.4cm}
\section{System Model and Power Allocation Problem}\label{sec:system model and secrecy rate}
Consider a single-antenna OFDM wiretap channel with $N$ subcarriers.
Let $H_i\in \mathbb{C}$ and $G_i\in \mathbb{C}$ represent the
complex channel coefficients of the $i$th subcarrier from the
transmitter to the legitimate receiver and to the eavesdropper,
respectively (see Fig. \ref{fig:SystemModel}(b)). The received
signals of the legitimate receiver and the eavesdropper can be
expressed as
\begin{subequations}\label{equ:OFDM normal signal}
\begin{eqnarray}
&y_i=H_i\sqrt{p_i}s_i+w_i,~~~~~i=1,\dots,N,\\
&z_i=G_i\sqrt{p_i}s_i+v_i,~~~~~~i=1,\dots,N,
\end{eqnarray}
\end{subequations}
respectively, where $s_i\in \mathbb{C}$ is the normalized channel
input signal with zero mean and unity variance; $p_i\geq 0$ denotes
the power of the $i$th subcarrier; $w_i\in \mathbb{C}$ and $v_i\in
\mathbb{C}$ are independent zero-mean circularly symmetric complex
Gaussian noises with unity variance at the legitimate receiver and
the
eavesdropper, respectively. 
{The channel coefficient $H_i$ is known to the transmitter and the legitimate receiver, and $G_i$ is known to the transmitter and the eavesdropper, which are satisfied if all the three nodes are within the same wireless system.} The channel inputs
$\{s_i\}$ are statistically independent and
identically distributed (i.i.d.) Gaussian signals or some practical discrete signals, e.g., QAM (see
Fig. \ref{fig:SystemModel}(a)).

\begin{figure}[t]
\centering
\includegraphics[scale=0.44]{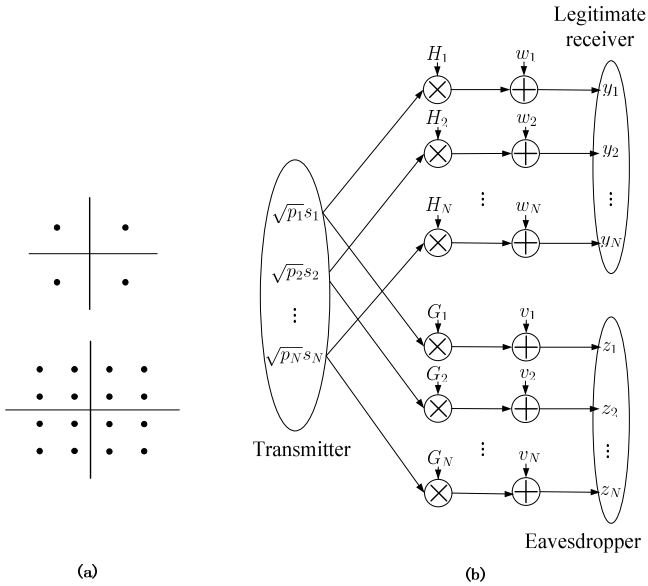}
\caption{(a) Finite discrete inputs: QPSK and 16QAM, (b) OFDM
wiretap channel.} \label{fig:SystemModel}
\end{figure}

The {secrecy rate} associated with the signal model in
(\ref{equ:OFDM normal signal}) can be shown to be \cite{Li2006}
\begin{equation}\label{Rs}
R_s(\textbf{p})\!=\frac{1}{N}\!\sum_{i=1}^{N}[\mathcal{I}(s_i;H_i\sqrt{p_i}s_i+w_i)\!\!-\!\!\mathcal{I}(s_i;G_i\sqrt{p_i}s_i+v_i)]^+,
\end{equation}
where $\textbf{p}=[p_1,p_2,...,p_N]^T$ contains all the subcarrier
power variables. Let us consider the following secrecy rate
maximization problem:
\begin{subequations}\label{problem1}
\begin{eqnarray}
R^\star = \underset{\mathbf{p}\succeq \mathbf{0}}{\max} &&\!\!\!\!\!
R_s(\textbf{p})
\\ 
~~~~~\textrm{s.t.}~ &&\!\!\!\!\! {\frac{1}{N}\sum\limits_{i=1}^N{p_i}\leq{P}},\label{eq2}
\end{eqnarray}
\end{subequations} where $P$ in \eqref{eq2} is the allowed maximal transmit power, and $R^\star$ denotes the maximum achievable secrecy rate.

\subsection{Gaussian Channel Inputs}
We first briefly review the case when $\{s_i\}$ are Gaussian
signals. In this case, the secrecy rate $R_s(\textbf{p})$ in
\eqref{Rs} can be easily reduced to
\begin{equation}\label{eq1}
R_{s}(\textbf{p})={\frac{1}{N}}
\sum_{i=1}^N[\;\log_2(1+{|H_i|^2p_i})-\log_2(1+{|G_i|^2p_i})\;]^+,
\end{equation}
which is known to be concave in $\mathbf{p}$. Therefore, the power
allocation problem \eqref{problem1} is a convex problem which can be
solved efficiently thanks to the following closed-form solution
\cite{Li2006,Liang2008,Jorswieck2008}:
\begin{eqnarray}\label{equ:gaussResult} p_i^\star\!=\!\!\begin{cases}\!\!
\frac{1}{2|H_i|^2|G_i|^2}\!\!\left[\!\sqrt{\!\!C_i^2\!\!-\!\frac{4|H_i|^2|G_i|^2(\lambda\!+\!|G_i|^2\!-\!|H_i|^2)}{\lambda}}\!-\!C_i\!\right],\\[0.1cm]
 ~~~~~~~~~~~~~~~~~~ ~~~~~~~~~~~~~{\textrm{if}~|H_i|^2\!-\!|G_i|^2\!>\!\lambda}\\[0.2cm]
 0,
 ~~~~~~~~~~~~~~~~~~~~~~~~~~~~~\textrm{otherwise},
\end{cases}
\end{eqnarray}
where $C_i=|H_i|^2+|G_i|^2$, and$\lambda\geq 0$ is the Lagrange multiplier associated with the
total power constraint \eqref{eq2}, and should be chosen such that
$\{p_i^\star\}$
in (\ref{equ:gaussResult}) satisfies ${ \frac{1}{N}}\sum_{i=1}^N{p_i^\star}={P}$. 
The reader can refer to \cite{Li2006,Liang2008,Jorswieck2008} for
more details.
\subsection{Discrete Channel Inputs}
In \cite{Rodrigues2010,Raghava2010,Bashar2011,Bashar2011arxiv,Li-Yates-Trappe2010},
it was observed from computer simulations that the secrecy rate
$R_s(\mathbf{p})$ in (\ref{eq1}) is nonconcave in $\mathbf{p}$ for
some discrete constellations, such as QPSK and 16QAM. One can infer
from Theorem 1 in \cite{Raghava2010} that $R_s(\mathbf{p})$ is
nonconcave for any uniformly distributed discrete inputs. {We now
prove that $R_s(\mathbf{p})$ is nonconcave for any discrete channel input distribution with a finite number of possible values:}
\begin{proposition}\label{proposition:secrecy rate}
Consider the following secrecy rate function:
\begin{equation}\label{equ:proposition1}
R(p)\triangleq
[\mathcal{I}(s;H\sqrt{p}s+w)-\mathcal{I}(s;G\sqrt{p}s+v)]^+,
\end{equation}
where $s$ has zero mean and unity variance, and $w$ and $v$ are
circularly symmetric complex Gaussian random variables with zero
mean and unity variance. Suppose that $s$ has a finite entropy,
i.e., $\mathcal{H}(s)<\infty$. Then, if $|H|>|G|$, $R(p)\geq0$ and
$R(p)$ is nonconcave in $p$; otherwise, $R(p)=0$ for all $p\geq0$.
\end{proposition}
\begin{proof}
We first show that $R({0})=\lim_{p\rightarrow\infty}R(p)=0$.
When $p=0$, one can easily show that
$\mathcal{I}(s;w)=\mathcal{I}(s;v)=0$, because $s$ is statistically independent of
$w$ and $v$. Thus, $R({0})=0$.
Since $s$ has a finite entropy, it
must be a discrete random variable. According to Lemma 6 of
\cite{Guo05}, we have
\begin{equation}\label{equ1inprop1}
\lim_{p\rightarrow\infty}\mathcal{I}(s;H\sqrt{p}s+w)=\lim_{p\rightarrow\infty}\mathcal{I}(s;G\sqrt{p}s+v)=
\mathcal{H}(s)<\infty.
\end{equation}
Therefore, \begin{equation}\label{eq8}
\lim_{p\rightarrow\infty}R(p)=0.
\end{equation}

Next, let us show that, when $|H|> |G|$, there exists a
$\hat{p}\in(0,\infty)$ such that $R(\hat{p})>0$. According to
\cite{Guo05}, the gradient of $\mathcal{I}(s;H\sqrt{p}s+w)$ is given
by
\begin{equation}\label{equ:derivative}
\frac{\partial{\mathcal{I}(s;H\sqrt{p}s+w)}}{\partial{p}}=|H|^2\mathrm{mmse}\left(|H|^2p\right)\geq
0,
\end{equation}
where 
\begin{eqnarray}\label{equ:MMSE}
&& \mathrm{mmse}(|H|^2p)\triangleq
\mathbb{E}\left[|s-\mathbb{E}(s|H\sqrt{p}{s}+w)|^2\right]
\end{eqnarray}
is the minimum mean square error (MMSE) of estimating $s$ with the
received signal $y= H\sqrt{p} s +w$. When $p$ equals zero,
$\mathrm{mmse}(|H|^2p)$ in (\ref{equ:MMSE}) attains its maximum
value, i.e., $\mathrm{mmse}(0)=E[|s|^2]=1$. Thus, by (\ref{equ:proposition1}) and
(\ref{equ:derivative}), it can be seen that $
R'(0)={|H|^2-|G|^2}>0$, which implies that there must exist a
positive $\hat{p}$ such that $R(\hat{p})>0$.

Since $R(p)$ is continuous and differentiable \cite{Guo05}, by the
Lagrange's mean value theorem \cite{Jeffreys1988} and the fact that
$R(\hat{p})>\lim_{p\rightarrow\infty}R(p)=0$, it can be inferred
that there must exist a point $\tilde{p}\in [\hat{p},\infty)$
satisfying $R(\tilde p) < \infty$ and $R'(\tilde{p})<0.$

Now suppose that $R(p)$ is a concave function on $p\in[0,\infty)$.
Then the following inequality
\begin{equation} \label{eq7}
R(p)\leq R(\tilde{p}) + R'(\tilde{p})(p-\tilde{p})
\end{equation}
must hold for any $p \in [0,\infty)$. By letting
$p\rightarrow\infty$ in \eqref{eq7}, we obtain
$\lim_{p\rightarrow\infty}R(p)=-\infty$ since $R'(\tilde p)<0$, which however contradicts
the fact of $\lim_{p\rightarrow\infty}R(p)=0$. Therefore, the
concavity assumption for $R(p)$ is not true.

When $|H|\leq |G|$, $\mathcal{I}(s;H\sqrt{p}s+w)\leq
\mathcal{I}(s;G\sqrt{p}s+v)$ for any $p\geq 0$. Hence, in this case,
$R(p)=0$ for all $p\geq0$. Proposition 1 is thus proved.
\end{proof}

Proposition \ref{proposition:secrecy rate} implies that finding an
optimal power allocation for maximizing $R_s(\mathbf{p})$ is a
non-trivial design problem, as will be discussed shortly. {By
Proposition \ref{proposition:secrecy rate}, the secrecy rate
function is not a concave function for any channel input
distribution with a finite entropy. 
Basically two types of inputs have finite entropy, including
discrete distributions with a finite number of distinct values and
some of the discrete distributions with a countable number of
distinct values\footnote{A discrete distribution with a countable
number of values may either have a finite entropy or an infinite
entropy {\cite[page 48]{Cover2006}}.}. On the other hand,
distributions with an infinite entropy basically include all
continuous distributions and some of the discrete distributions with
a countable number of values, and this type of channel inputs may
have a concave secrecy rate function.}

Let us provide some numerical results to illustrate Proposition
\ref{proposition:secrecy rate}. Figure \ref{fig:I_input} shows the
mutual information $\mathcal{I}(s;H\sqrt{p}s+w)$,
$\mathcal{I}(s;G\sqrt{p}s+v)$, and the secrecy rate $R(p)$ for
Gaussian and QPSK channel inputs. One can observe from Fig.
\ref{fig:I_input}(a) and Fig. \ref{fig:I_input}(b) that
$\mathcal{I}(s;H\sqrt{p}s+w)$, $\mathcal{I}(s;G\sqrt{p}s+v)$ and
$R(p)$ are all concave when the inputs are Gaussian. From Fig.
\ref{fig:I_input}(c), $\mathcal{I}(s;H\sqrt{p}s+w)$ and
$\mathcal{I}(s;G\sqrt{p}s+v)$ are also concave for QPSK channel
inputs; however, the associated $R(p)$ is obviously not concave, as
shown in Fig. \ref{fig:I_input}(d).

\begin{figure}[t]
\centering
\includegraphics[scale=0.45]{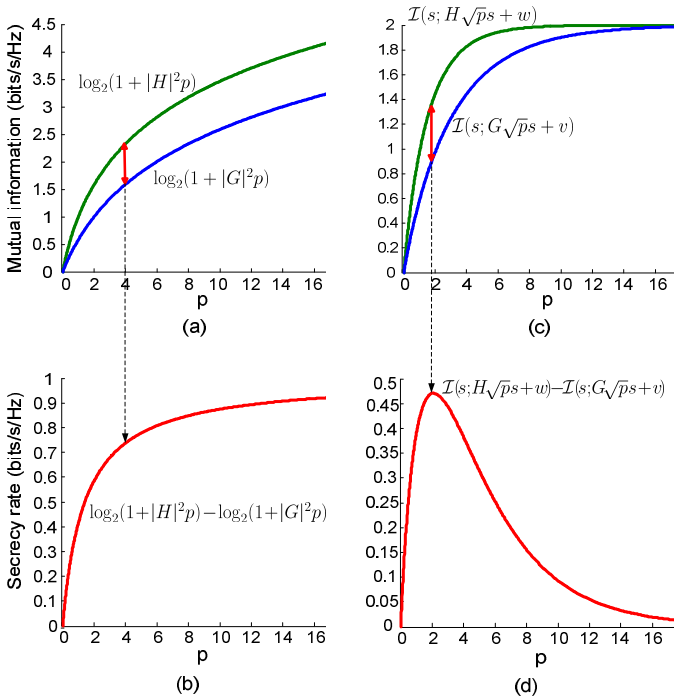}
\caption{(a) Mutual information and (b) secrecy rate for a Gaussian
channel input. (c) Mutual information and (d) secrecy rate for a
QPSK channel input. {Here $|H|=1$ and $|G|=0.5$}.}
\label{fig:I_input}
\end{figure}

\begin{figure}[t]
\centering
\includegraphics[scale=0.45]{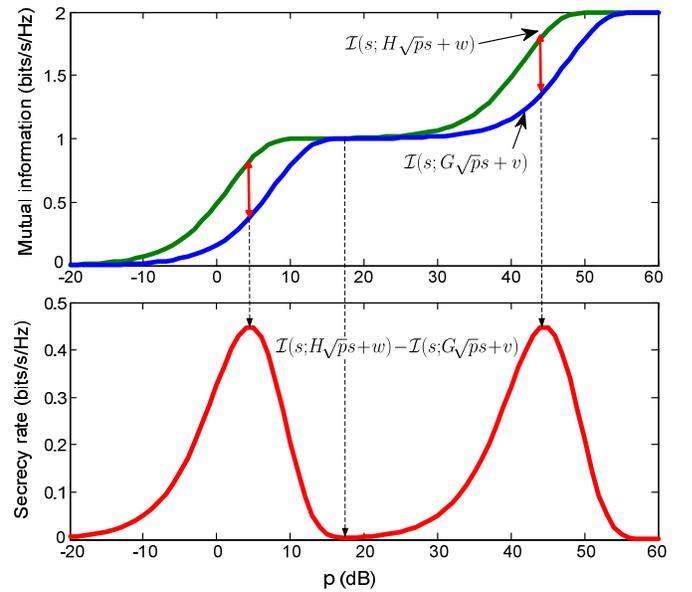}
\caption{Secrecy rate for a {non-standard} 4-PAM channel input with
probability mass function given by (\ref{dis:mess}). Here {$|H|=1$
and $|G|=0.5$}.} \label{fig:I_special distributed}
\end{figure}

We further claim that the secrecy rate $R(p)$ is not even an
\emph{quasi-concave} function for certain channel input distributions, by showing a counter example that the secrecy rate may have multiple peaks. Figure
\ref{fig:I_special distributed} further illustrates the secrecy rate
$R(p)$ in the logarithmic scale of $p$ for a {non-standard} 4-PAM channel
input $s$ \footnote{While the mutual information
$\mathcal{I}(s;H\sqrt{p}s+w)$ and $\mathcal{I}(s;G\sqrt{p}s+v)$ are
concave functions in $p$, they appear to be nonconcave in
logarithmic scale.}, where the probability mass function of $s$ is
given by
\begin{equation}\label{dis:mess}
P_{s}\sim
\begin{bmatrix}
&\!\!\!\!-51q & -50q & 50q & 51q~ \\
&\!\!\!\!0.25 & 0.25 &0.25 &0.25~
\end{bmatrix},
\end{equation}
where $q$ is a normalization parameter used to maintain the unity
variance of $s$. {It is interesting to see from Fig.
\ref{fig:I_special distributed} that the secrecy rate $R(p)$ has two
peaks. Hence, $R(p)$ is neither concave nor quasi-concave. When
$p<-10 \mathrm{dB}$, both Bob and Eve cannot identify the
constellation points. Since the constellation of $s$ in
(\ref{dis:mess}) has two groups (i.e., the group of $\{-51q,-50q\}$
and the group $\{51q,50q\}$), Bob and Eve start to be able to
identify the two groups as $p$ increases. Since $|H|>|G|$,
$I(s;H\sqrt{p}s+w)$ is larger than $I(s;G\sqrt{p}s+v)$, and so
$R(p)>0$. When $p\thickapprox 18 \mathrm{dB}$, both Bob and Eve can
identify the two groups of $s$, and  $I(s;H\sqrt{p}s+w)\thickapprox
I(s;G\sqrt{p}s+v) \thickapprox 1~\mathrm{bits/s/Hz}$, and so $R(p)$
decreases to nearly 0. When $p\geq 20 \mathrm{dB}$, Bob and Eve
start to identify each constellation point. Since Bob has better
channel quality, $R(p)$ increases with $p$ again.} This example also
shows that the conjecture in \cite{Raghava2010}, which claims the
secrecy rate under discrete finite constellations has a single
maximum, is not true for some discrete channel inputs.

\section{Proposed Power Allocation Algorithm}\label{sec:power allocation}
We now present {a computationally efficient} Lagrange
dual optimization algorithm to handle the nonconvex problem
\eqref{problem1}. {We will show that the
proposed algorithm yields a power allocation solution for which the
duality gap decreases with $N$ in the order of $O(1/\sqrt{N})${,
provided that the channel has a finite delay spread}. This mild condition is satisfied in practical OFDM systems.}

\subsection{Asymptotic Optimal Power Allocation by Dual
Optimization}\label{sec:power allocation-A} The Lagrangian of
problem (\ref{problem1}) is given by
\begin{eqnarray}\label{eq6}
L(\textbf{p},\lambda)=\!\!\!\!\!\!\!\!&&\!{\frac{1}{N}}\!\sum_{i=1}^N[\mathcal{I}(s_i;H_i\sqrt{p_i}s_i\!\!+\!\!w_i)\!-
\!\mathcal{I}(s_i;G_i\sqrt{p_i}s_i+v_i)]^+\!\nonumber\\
&&+\!\lambda\!\left(\!\!P-{\frac{1}{N}}\sum_{i=1}^Np_i\!\!\right),
\end{eqnarray}
where $\lambda\geq 0$ is the dual variable associated with the
constraint \eqref{eq2}. The dual problem is given by
\begin{equation}\label{pro:dual}
\begin{array}{ll}
D^\star =&\underset{\lambda\geq 0}\min~~d(\lambda),~~~~~~~
\end{array}
\end{equation}
where $D^\star$ denotes the optimal dual objective value, and
$d(\lambda)$ is the dual function given by
\begin{eqnarray}\label{pro:inner prob.}
d(\lambda)\triangleq \underset{\textbf{p}\succeq
0}\max{~~L(\textbf{p},\lambda)}.
\end{eqnarray}

The Lagrange dual method first solves problem \eqref{pro:inner
prob.} for a given dual variable $\lambda$. According to
\eqref{eq6}, problem \eqref{pro:inner prob.} can be decomposed into
$N$ separate subproblems, i.e.,
\begin{eqnarray}\label{equ:dlambda}
d(\lambda)={\frac{1}{N}}\sum_{i=1}^N B(p_i;\lambda,H_i,G_i) +
\lambda P,
\end{eqnarray}
where
\begin{eqnarray}\label{eq4}
\!\!\!\!\!\!\!\!\!\!\!\!\!\!&&\!\!\!\!\!\!\!\!\!\!\!B(p_i;\lambda,H_i,G_i)\!\nonumber\\
\!\!\!\!\!\!\!\!\!\!\!\!\!=&&\!\!\!\!\!\!\!\!\!\!\!\max_{p_i\geq
0}\left[\mathcal{I}(s_i;H_i\sqrt{p_i}s_i\!\!+\!\!w_i)\!\!-\!\mathcal{I}(s_i;G_i\sqrt{p_i}s_i\!\!+\!\!v_i)\right]^+\!\!\!-\!\lambda
p_i~~
\end{eqnarray}
is a one-dimensional power allocation subproblem for subcarrier $i$,
which can be efficiently solved by simple line search \cite{Boyd2003}. 

Notice that $d(\lambda)$ is a convex function \cite{Boyd2003} and
its subgradient can be easily seen, from {\eqref{equ:dlambda} and
\eqref{eq4}}, to be
$$\triangledown d(\lambda) = P -{\frac{1}{N}}\sum_{i=1}^{N} p_i^\star,$$ where $p_i^\star$ is the
optimal solution to \eqref{eq4}. Therefore, the dual variable
$\lambda$ can be efficiently updated using the bisection method
\cite{Boyd2003}. If $P -{\frac{1}{N}}\sum_{i=1}^{N} p_i^\star>0$,
then the subgradient $\triangledown d(\lambda)>0$, and thus we
decrease $\lambda$ in the bisection method for finding $D^\star$ given by (\ref{pro:dual}); otherwise we increase
$\lambda$. The resulting power allocation algorithm, called
Algorithm 1, for finding the desired $\mathbf{p}^\star$ for problem
\eqref{problem1} is summarized in Table \ref{tab1}. {It is important
to note from (\ref{eq8}) that the secrecy rate is not an increasing
function of the transmit power. Therefore, when the total power $P$
is large enough, the optimal total transmit power
${\frac{1}{N}}\sum_{i=1}^N{p_i^\star}$ by solving (\ref{eq4}) can be
strictly lower than the available transmit power $P$, and the
optimal dual variable is $\lambda^\star=0$ by the Karush-Kuhn-Tucker
(KKT) conditions \cite{Boyd2003}.
It is worthwhile to mention
that Algorithm 1 can also be applied to the case of Gaussian inputs,
where the per-subcarrier power allocation subproblem (\ref{eq4}) has
a closed-form solution exactly the same as (\ref{equ:gaussResult})
\cite{Li2006,Liang2008,Jorswieck2008}.


\begin{table} \caption{} \label{tab1} \centering
\begin{tabular}{l}
\hline Algorithm 1: Proposed power allocation scheme for discrete inputs.\\
\hline \textbf{Given:} ~~~~$\lambda_h\geq\lambda_l=0$, convergence tolerance $\varepsilon$\\
\textbf{repeat:}\\ ~~step 1: ~~~update
$\lambda=\frac{1}{2}(\lambda_l+\lambda_h)$\\~~step 2:
~~~obtain $\{p_i\}_{i=1}^N$ by solving problem (\ref{eq4})\\
~~step 3:~~~~if~ ${\frac{1}{N}}\sum\limits_{i=1}^{N}p_i<P$, then
update
$\lambda_h=\lambda$, else update $\lambda_l=\lambda$\\
\textbf{until:}~~~~~~~$\lambda_h-\lambda_l<\varepsilon$\\
\textbf{output:}~~~~output $\lambda^\star = \lambda_l$.\\
\hline
\end{tabular}
\end{table}

Since the primal problem (\ref{problem1}) is nonconvex, there is a
duality gap between the optimal $R^\star$ and optimal $D^\star$,
i.e., $D^\star-R^\star>0$, where $R^\star$ and $D^\star$ are defined
in \eqref{problem1} and \eqref{pro:dual}, respectively. Under a
Lipschitz continuity assumption on the channel coefficients
\cite{Luo2009}, it can be shown that the duality gap between
$D^\star$ and $R^\star$ diminishes with $N$ in the order of
$O({1}/{\sqrt{N}})$. Prior to presenting such result, some
quantities used in the discrete frequency domain and the
corresponding quantities in the continuous frequency domain need to be
reviewed. Let $H(f)$ and $G(f)$ denote the frequency domain channel
responses of the legitimate receiver and the
eavesdropper, respectively. 
Owing to uniform sampling in
Discrete Fourier transform (DFT) over the normalized frequency
interval $0\leq f\leq 1$, we have
\begin{equation}
H(i/N)=H_i,~ G(i/N)=G_i,~1\leq i \leq N.
\end{equation}

{In general, the time-domain channel has a finite delay spread
(i.e., it is nonzero only on a finite time interval), leading to the
fact that any order derivatives of its frequency response
exist and are bounded \cite[pp. 94-96]{Papoulis1977}. {In particular,} the first
derivatives of $H(f)$ and $G(f)$ exist and are bounded. Therefore,
there exist some $L_H,L_G>0$ such that the following Lipschitz
continuous conditions hold for all $f,f'\in[0,1]$:
\begin{eqnarray} \label{equ:hi_lipschitz}
\big|H(f)\!\!-\!\!H(f')\big|\!\!\leq \!L_H\! |f\!-\!f'|,~
\big|G(f)\!\!-\!\!G(f')\big|\!\!\leq\!L_G\!|f-f'|
\end{eqnarray}
We can show the following proposition:}

\begin{proposition}\label{proposition:approximate optimal}
{Suppose that the channel coefficients $H(f)$ and $G(f)$ are
Lipschitz continuous satisfying (\ref{equ:hi_lipschitz}).
Then
\begin{equation}\label{equ:pro1-result}
0\leq D^\star-R^\star\leq O\left(\frac{1}{\sqrt{N}}\right),
\end{equation}
where $R^\star$ and $D^\star$ are defined in problem
(\ref{problem1}) and problem (\ref{pro:dual}), respectively.}
\end{proposition}

\begin{proof} In accordance with
\cite[Theorem 2]{Luo2009}, for proving \eqref{equ:pro1-result} it is
sufficient to show that there exists a constant $L_{R}>0$ such that
the difference between the secrecy rates on any two frequencies,
$R_{s}(f,p)$ and $R_{s}(f',p')$, is bounded, i.e.,
\begin{equation}\label{key}
|R_{s}(f,p)-R_{s}(f',p')|\leq L_{R}\left(|f-f'|+|p-p'|\right),
\end{equation} for any $f,f' \in [0,1]$ and $p,p'\geq 0$, namely,
that the secrecy rate $R_{s}(f,p)$ is Lipschitz continuous. The
proof of \eqref{key} is presented in Appendix I.
\end{proof}


\subsection{Numerical Results}\label{sec:NRofPA}
We now provide some numerical results to illustrate the efficacy of
the proposed power allocation scheme (Algorithm 1 in Table
\ref{tab1}). Suppose that the OFDM system has $N = 64$ subcarriers.
The length of the cyclic prefix (CP) is $N_{cp} = 16$ channel
samples. {Each channel realization is composed of 8 i.i.d. Rayleigh
fading paths with the maximum time delay of 7 samples and the
channel power normalized to unity.} The numerical results are
obtained by averaging over 300 channel realizations for each
signal-to-noise ratio (SNR) value, where $\mathrm{SNR}={P}$ (which
is the ratio of the available transmit power to AWGN power
($\mathbb{E}[|w_i|^2]=\mathbb{E}[|v_i|^2]=1$) per subcarrier). The
proposed power allocation algorithm is compared with two existing
algorithms. The first one is the equal power allocation scheme,
denoted by ``equal PA'', which allocates the total power uniformly
over all the subcarriers that satisfy $|H_i|>|G_i|$; the second one
is the power allocation scheme for Gaussian channel inputs, which is
given by (\ref{equ:gaussResult}), denoted by ``PA of
(\ref{equ:gaussResult})''.
\begin{figure}[t]
\centering
\includegraphics[scale=0.44]{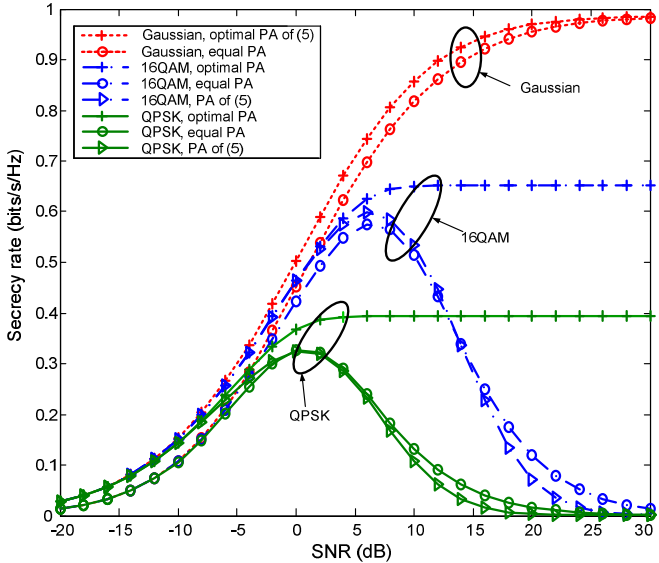}
\caption{Secrecy rates achieved by Algorithm 1 (in Table \ref{tab1})
denoted by ``$+$", equal power allocation scheme denoted by
``$\circ$", and power allocation scheme given by
(\ref{equ:gaussResult}) denoted by ``$\triangleright$".}
\label{Fig:comparison}
\end{figure}

\begin{figure}[t] \centering
\includegraphics[scale=0.45]{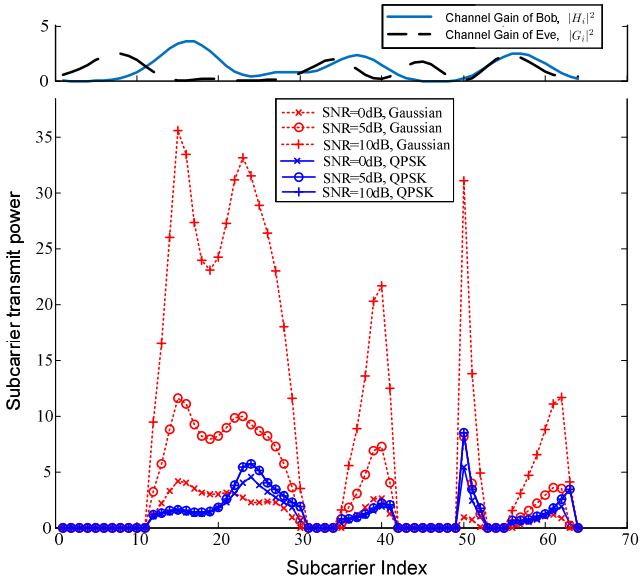}
\caption{Power allocation results obtained using Algorithm 1 for the
case of  Gaussian inputs and the case of QPSK inputs, where
$\mathrm{SNR} = 0\mathrm{dB}$; $5\mathrm{dB}$; $10\mathrm{dB}$ and
$N=64$. } \label{fig:power_to_sub}
\end{figure}

Figure \ref{Fig:comparison} illustrates the achievable secrecy rates
of the three power allocation schemes for different channel inputs.
It can be observed from this figure that Algorithm 1 (denoted by
``$+$'') achieves the highest secrecy rate (although it yields the
same optimal solution (\ref{equ:gaussResult}) for Gaussian channel
inputs). One can observe, from Fig. \ref{Fig:comparison}, that under
discrete inputs (QPSK and 16QAM), both the equal PA scheme (denoted
by ``{\Large $\circ$}'') and the PA scheme of
(\ref{equ:gaussResult}) (denoted by ``{\Large$\triangleright$}'') do
not perform well---their secrecy rates even drop to zero in the high
SNR regime. {The secrecy rate utilized by Algorithm 1 increases with
SNR, but saturates (rather than drops to zero) in the high SNR
regime, implying that there exists a power threshold for which the
secrecy rate reaches the maximum, and the power exceeding the
threshold will not improve the secrecy rate any more. The reason for
this is that the secrecy rate is not an increasing function of the
transmit power due to (\ref{eq8}).} {Similar observations were also reported in
\cite{Bashar2011arxiv,Rodrigues2010}.}

{Figure \ref{fig:power_to_sub} illustrates the power allocation
results obtained by Algorithm 1 for Gaussian and QPSK channel
inputs. Three SNR values are considered, i.e., $\mathrm{SNR} =
0\mathrm{dB}$, $5\mathrm{dB}$, $10\mathrm{dB}$. It can be observed,
from Fig. \ref{fig:power_to_sub}, that the transmit power is only
allocated to the subcarriers on which Bob has larger channel gain
than Eve. As previously mentioned, for discrete channel inputs, it
may not be true that the total power constraint (\ref{eq2}) is
active for the power allocation result obtained by Algorithm 1.
Actually, when $P$ is larger than a threshold, the power allocation
result remains unchanged. As shown in Fig. \ref{fig:power_to_sub},
when the channel inputs are QPSK, the power allocation result
associated with $\mathrm{SNR} = 5\mathrm{dB}$ almost coincides with
that associated with $\mathrm{SNR} = 10\mathrm{dB}$. This also
demonstrates that the total transmit power constraint (\ref{eq2}) is
inactive when $\mathrm{SNR} = 10\mathrm{dB}$ for QPSK inputs.}

\section{Joint Signal and Artificial Noise Design}\label{sec:AN}
{As presented in the previous section, the secrecy rate obtained by
Algorithm 1 reaches the maximum value and saturates in the high SNR
regime. {In other words, {the transmitter in the high SNR regime will not consume all the available transmit power}.} This motivates the use of the remaining
transmit power for generating artificial noise to jam the
eavesdropper, thus further increasing the secrecy rate. In
the subsequent subsections, the independent frequency-domain artificial noise
design is shown to be ineffective first. Then a novel
time-domain artificial noise scheme is proposed to exploit temporal degrees of freedom provided by the cyclic prefix of OFDM systems to jam the eavesdropper and boost the secrecy rate.}


\begin{figure*}
\centering
\scalebox{0.8}{\includegraphics{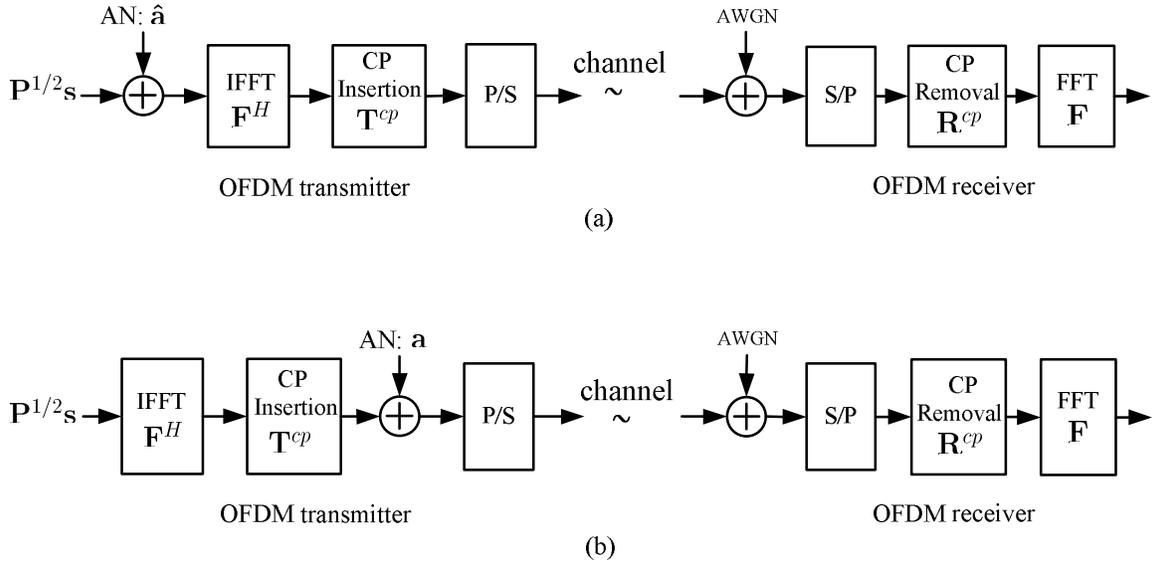}}\\
\caption{(a) System model of OFDM transceiver with artificial noise
(AN) added in the frequency domain. (b) System model of OFDM
transceiver with artificial noise added in the time domain.}
  \label{fig:OFDM}
\end{figure*}
\subsection{{Ineffectiveness} of {Independent} Frequency-Domain Artificial Noise Design}\label{sec:naifreAN}
Let us first consider a naive strategy by adding artificial
noise to all the subcarriers in the frequency domain, as illustrated
in Fig. \ref{fig:OFDM}(a), where $\mathbf{P} = \mathrm{Diag}({p_1},
{p_2},..., {p_N})$ and $\mathbf{s}=[s_1,s_2,...,s_N]^T$. The
corresponding received signals of the legitimate receiver and the
eavesdropper on the $i$th subcarrier can be expressed as
\begin{subequations}\label{equ:FD-AN}
    \begin{equation}
    y_i=H_i(\sqrt{p_i}s_i+\hat{a}_i)+w_i    ,~~~~~i=1,\dots,N,
    \end{equation}
    \begin{equation}
    z_i=G_i(\sqrt{p_i}s_i+\hat{a}_i)+v_i,~~~~~~i=1,\dots,N,
    \end{equation}
\end{subequations}
where $\hat{a}_i \in \mathbb{C}$ is the artificial noise term added
to the $i$th subcarrier. It is assumed that {$\hat{a}_i$ {are statistically independent} across the subcarriers and} $\hat{a}_i\sim
\mathcal{CN}(0,\sigma_{a,i}^2)$, where $\sigma_{a,i}^2\geq 0$ is the
artificial noise power for subcarrier $i$. According to
\cite{Li2006}, the secrecy rate for the signal model
(\ref{equ:FD-AN}) is given by
\begin{align}\label{Rs_ANinFdomain}
R_s^{AN}\nonumber
=&{\frac{1}{N}}\!\!\sum_{i=1}^{N}\big[\mathcal{I}\big(s_i;H_i(\sqrt{p_i}s_i+\hat{a}_i)+w_i\big)\!\!\nonumber\\
&~~~~~~~~-\!\!\mathcal{I}\big(s_i;G_i(\sqrt{p_i}s_i+\hat{a}_i)+v_i\big)\big]^+.
\end{align}
In order to maximize $R_s^{AN}$ in (\ref{Rs_ANinFdomain}), power
parameters $\{p_i\}$ and $\{\sigma_{a,i}^2\}$ should be jointly
optimized, which can be formulated as the following optimization
problem:
\begin{subequations}\label{pro:fre-an}
\begin{eqnarray}
\!\!\!\!\!\!\!\!\underset{\{p_i\}_{i=1}^N,\{\sigma_{a,i}^2\}_{i=1}^N}{\max}\!\!
&&\!\!\!\!\!\! R_s^{AN}\\
\textrm{s.t.}~~~~~ &&\!\!\!\!\!\! {\frac{1}{N}}\sum\limits_{i=1}^N(p_i+\sigma_{a,i}^2)\leq{P}~~~~~\label{equ1-in-proof}\\
&&\!\!\!\!\!\!\sigma_{a,i}^2\geq 0,~{p_i}\geq 0,~i =1,2,...,N.\label{equ2-in-proof}
\end{eqnarray}
\end{subequations}
The optimal solution for artificial noise power parameters of
\eqref{pro:fre-an} is given in the following proposition.
\begin{proposition}\label{property:frequency AN}
For {any given} distribution of $\{s_i\}$, the optimal artificial
noise power $\sigma_{a,i}^2$ to problem (\ref{pro:fre-an}) is
equal to zero for all $i$.
\end{proposition}
{For the case that $\{s_i\}$ are Gaussian, Proposition \ref{property:frequency AN} can be simply proved by showing that, when $|H_i|^2>|G_i|^2$, the secrecy rate $R_s^{AN}$ is strictly decreasing in $\sigma_{a,i}^2$ for any $p_i\geq 0$. The proof for the case of arbitrarily distributed $\{s_i\}$ is more complicated, because $R_s^{AN}$ may be non-monotonic in $\sigma_{a,i}^2$. The proof details for general distributions of $\{s_i\}$ are presented in Appendix II.}
%

Proposition \ref{property:frequency AN} implies that adding
artificial noise in the frequency domain is not effective since it
only degrades the achievable secrecy rate. Next let us turn to the
design of artificial noise in the time domain which is able to boost
the secrecy rate. \vspace{-0.3cm}
\subsection{Proposed Time-domain Artificial Noise Design}\label{sec:T-AN}
The proposed wiretap OFDM system with a time-domain artificial noise
is
illustrated in Fig. \ref{fig:OFDM}(b). 
As a standard OFDM system, at first the frequency domain signal
$\mathbf{P}^{1/2} \mathbf{s}$ is transformed to the time domain by
inverse fast Fourier transform (IFFT) and then the cyclic prefix
(CP) is inserted. Then, an artificial noise term is added to this
time-domain signal before transmission. The receiver will discard
the CP and then transform the remaining signal to the frequency
domain by FFT. All of these operations can be expressed
by linear matrix operations \cite{Wang2000}.

\newcounter{TempEqCnt1}
\setcounter{TempEqCnt1}{\value{equation}}
\setcounter{equation}{31}
\begin{figure*}[b]
\hrulefill
\begin{equation}\label{equ:A-Lagrange}
\begin{array}{ll}
\!\!\!\!\!\!\!L(\mathbf{p},\mathbf{\Sigma}_d,\lambda)\!\!=\!\!\!\!\!\!&{\frac{1}{N}}\sum\limits_{i=1}^{N}\left[\mathcal{I}(s_i;\!
H_i\sqrt{p_i}s_i\!+\!w_i)\!-\!\mathcal{I}(s_i;\!\frac{G_i\sqrt{p_i}s_i}{\sqrt{\mathbf{b}_i^H\mathbf{\Sigma}_d\mathbf{b}_i\!+\!1}}
\!+\!\hat{v}_i)\right]^++\lambda\left(P-{\frac{1}{N}}\sum\limits_{i=1}^{N}{p_i}-{\frac{1}{N}}\mathrm{tr}(\mathbf{\Sigma}_d)\right)
\end{array}
\end{equation}
\end{figure*}
\setcounter{equation}{\value{TempEqCnt1}}

\newcounter{TempEqCnt}
\setcounter{TempEqCnt}{\value{equation}}
\setcounter{equation}{34}
\begin{figure*} [b]
\begin{equation}\label{pro:EA-dualfuc-pi}
p_i^\star=\arg \underset{p_i>0}\max~[\mathcal{I}(s_i;
H_i\!\sqrt{p_i}s_i\!+w_i\!)\!-\!\mathcal{I}(\!s_i;\!\!\frac{G_i\sqrt{p_i}s_i}{\sqrt{\mathbf{b}_i^H\mathbf{\Sigma}_d\mathbf{b}_i\!\!+\!\!1}}
+\hat{v}_i\!)]^+\!-\!\lambda {p_i},~i=1,2,...,N.
\end{equation}
\end{figure*}
\setcounter{equation}{\value{TempEqCnt}}

Let $\mathbf{F}$ and $\mathbf{F}^H$ denote  the $N\times N$ FFT and
IFFT matrices, and let $N_{cp}$ denote the length of CP. The
matrices for CP insertion and removal are represented by
$\mathbf{T}^{cp}=[\mathbf{\tilde{E}}^T_{N_{cp}\times N}
~\mathbf{I}_N]^T$ and $\mathbf{R}^{cp}=[\mathbf{0}_{N\times N_{cp}}
~\mathbf{I}_N]$, respectively, where
$\mathbf{\tilde{E}}_{N_{cp}\times N}$ contains the last $N_{cp}$
rows of the $N\times N$ identity matrix $\mathbf{I}_N$. Let $[h(0),
h(1),..., h(L)]$ and $[g(0),g(1),...,g(L)]$ represent the
time-domain channel impulse responses from the transmitter to the
legitimate receiver and the eavesdropper, respectively, where
$L<N_{cp}$ is the maximum delay spread. Then, following the system
block diagram in Fig. \ref{fig:OFDM}(b), the received signals of the
legitimate receiver and the eavesdropper can be expressed as
\cite{Wang2000}
\begin{subequations}\label{equ:S+AN}
\begin{equation}
\begin{array}{ll}
\mathbf{y}&=\mathbf{FR}^{cp}\mathbf{H}_0(\mathbf{T}^{cp}\mathbf{F}^H\mathbf{P}^{1/2}\mathbf{s}+\mathbf{a})+\mathbf{w}\\
&=\mathbf{H}\mathbf{P}^{1/2}\mathbf{s}+\mathbf{FR}^{cp}\mathbf{H}_0\mathbf{a}+\mathbf{w},
\end{array}
\end{equation}
\begin{equation}
\begin{array}{ll}
\mathbf{z}&=\mathbf{FR}^{cp}\mathbf{G}_0(\mathbf{T}^{cp}\mathbf{F}^H\mathbf{P}^{1/2}\mathbf{s}+\mathbf{a})+\mathbf{v}\\
&=\mathbf{G}\mathbf{P}^{1/2}\mathbf{s}+\mathbf{FR}^{cp}\mathbf{G}_0\mathbf{a}+\mathbf{v},
\end{array}
\end{equation}
\end{subequations}
where $\mathbf{a}\in \mathbb{C}^{N+N_{cp}}$ is a zero-mean complex
Gaussian random vector, $\mathbf{H}_0\in\mathbb{C}^{(N+N_{cp})\times(N+N_{cp})}$ is a {Toeplitz} channel matrix given by
\begin{equation}\nonumber \mathbf{H}_0=
\begin{bmatrix}
&h(0) &0 &0 &\cdots &0\\\vspace{-0.2cm} &\vdots &h(0) &0 &\cdots
&0\\\vspace{-0.2cm} &h(L) &\cdots &\ddots &\cdots
&\vdots\\\vspace{-0.1cm} &\vdots &\ddots &\cdots &\ddots &0\\
&0&\cdots &h(L) &\cdots &h(0)
\end{bmatrix},
\end{equation}
and so is the matrix $\mathbf{G}_0$, i.e., by replacing $h(l)$ with
$g(l), l=0,1,...,L$, in $\mathbf{H}_0$. Moreover, in
\eqref{equ:S+AN},
$\mathbf{H}=\mathbf{FR}^{cp}\mathbf{H}_0\mathbf{T}^{cp}\mathbf{F}^H=\mathrm{{Diag}}(
H_1,H_2,...,H_N)$ and
$\mathbf{G}=\mathbf{FR}^{cp}\mathbf{G}_0\mathbf{T}^{cp}\mathbf{F}^H
=\mathrm{{Diag}}(G_1,G_2,...,G_N)$, in which $H_i$ and $G_i$ are the
frequency responses of the legitimate receiver's channel and the
eavesdropper's channel, respectively, and
$\mathbf{{w}}\sim\mathcal{CN}(\mathbf{0},\mathbf{I}_N)$ and
$\mathbf{{v}}\sim\mathcal{CN}(\mathbf{0},\mathbf{I}_N)$ are the
corresponding noise vectors at the legitimate receiver and the
eavesdropper, respectively.

The received signal vectors given in (\ref{equ:S+AN}) can be
equivalently written as
\begin{subequations}\label{equ:s+AN-subcarrier}
\begin{equation}
y_i=H_i\sqrt{p_i}s_i+\mathbf{f}_i^T\mathbf{R}^{cp}\mathbf{H}_0\mathbf{a}+w_i,~~~i=1,2,...,N,
\end{equation}
\begin{equation}
z_i=G_i\sqrt{p_i}s_i+\mathbf{f}_i^T\mathbf{R}^{cp}\mathbf{G}_0\mathbf{a}+v_i,~~~i=1,2,...,N,
\end{equation}
\end{subequations}
respectively, where $\mathbf{f}_i^T\in \mathbb{C}^N$ is the $i$th
row of the FFT matrix $\mathbf{F}$. Note that this corresponds to a
simplified model of secure design in MIMO secrecy networks with a
cooperative jammer \cite{Fakoorian2011}.

In order not to interfere with the legitimate receiver, the
artificial noise is fully laid in the null space of the channel of
the legitimate receiver. Specifically, we let
\begin{equation}\label{set of ZF-designed AN}
\!\!\!\!\mathbf{a}=\mathbf{U}\mathbf{d},
\end{equation}
where $\mathbf{U}$ is a semi-unitary matrix whose column vectors
span the null space of $\mathbf{R}^{cp}\mathbf{H}_0$, i.e.,
\begin{equation}\label{equ: AN null codition}
\begin{array}{ll}
\mathbf{R}^{cp}\mathbf{H}_0\mathbf{U}=\mathbf{0},~~
\mathbf{U}^H\mathbf{U}=\mathbf{I}_{N_{cp}},
\end{array}
\end{equation}
and $\mathbf{d}\sim\mathcal{CN}(\mathbf{0},\mathbf{\Sigma}_d)$ in
which $\mathbf{\Sigma}_d$ is the covariance matrix of the artificial
noise vector $\mathbf{d}$ to be determined. {As the dimension of
$\mathbf{R}^{cp}\mathbf{H}_0$ is $N\!\times\! (N\!\!+\!\!N_{cp})$,
the dimension of $\mathbf{U}$ is $(N\!\!+\!\!N_{cp})\!\times\!
N_{cp}$ and the dimension of $\mathbf{d}$ is $N_{cp}$.}

Now, by (\ref{set of ZF-designed AN}) and (\ref{equ: AN null
codition}), the received signal in (\ref{equ:s+AN-subcarrier})
reduces to
\begin{subequations}\label{sigmod:sign-an-null}
\begin{eqnarray}
\!\!\!\!\!\!\!\!\!\!\!\!\!\!\!\!\!\!\!\!\!\!\!\!\!\!\!\!&&y_i=H_i\sqrt{p_i}s_i+w_i,~i=1,2,...,N,\\
\!\!\!\!\!\!\!\!\!\!\!\!\!\!\!\!\!\!\!\!\!\!\!\!\!\!\!\!&&z_i=G_i\sqrt{p_i}s_i+\mathbf{f}_i^T\mathbf{R}^{cp}\mathbf{G}_0\mathbf{U}\mathbf{d}+v_i,~i=1,2,...,N.
\end{eqnarray}
\end{subequations}
In general, the receiver detects the information in a
subcarrier-by-subcarrier manner, and the secrecy rate achieved by
(\ref{sigmod:sign-an-null}) is given by \cite{Li2006}
\begin{equation}
\small
{\frac{1}{N}}\sum\limits_{i=1}^{N}\left[\!\mathcal{I}(s_i;
H_i\sqrt{p_i}s_i+{w}_i)-\mathcal{I}(s_i;\frac{G_i\sqrt{p_i}s_i}{\sqrt{\mathbf{b}_i^H\mathbf{\Sigma}_{d}\mathbf{b}_i+
1}} +\hat{v}_i)\right]^+,
\end{equation}
where
$\mathbf{b}_i^H\triangleq\mathbf{f}_i^T\mathbf{R}^{cp}\mathbf{G}_0\mathbf{U}$,
and $\hat{v}_i\sim \mathcal {CN}(0,1)$. The joint power allocation
and artificial noise design problem can be formulated as
\begin{subequations}\label{pro:main}
\small
\begin{eqnarray}
\!\!\!\!\!\!\!\!\!\!\!\!\!\!\!\!\!\!\!\!\underset{\{\!p_i\!\}_{\!i=1}^N,\mathbf{\Sigma}_d}\max \!\!\!\!\!&&\!\!\!\!\!
{\frac{1}{N}}\sum\limits_{i=1}^{N}\!\!\left[\!\mathcal{I}(\!s_i;\!
H_i\sqrt{p_i}s_i\!\!+\!\!w_i\!)\!\!-\!\!\mathcal{I}(\!s_i;\!\frac{G_i\sqrt{p_i}s_i}{\sqrt{\mathbf{b}_i^H\mathbf{\Sigma}_{d}\mathbf{b}_i\!\!+\!
1}}\!\!+\!\hat{v}_i\!)\!\right]^+~~~\\
\textrm{s.t.}~~~&&\!\!\!\!{\frac{1}{N}}\left(\sum\limits_{i=1}^{N}{p_i}+\mathrm{tr}(\mathbf{\Sigma}_d)\right)\leq
P\label{const:ANmain_2}\\[0.15cm]
&&\!\!\!\!\mathbf{\Sigma}_d\succeq0,~~ p_i\geq0,~~ i=1,\ldots,N.~~~
\end{eqnarray}
\end{subequations}
\subsection{Lagrange Dual Optimization to Problem (\ref{pro:main}) through Successive Convex Approximation}\label{sec:JPA-AN-dual method}
Problem (\ref{pro:main}) is nonconvex and difficult to handle.
Again, as in Section \ref{sec:power allocation}, we consider
Lagrange dual optimization to solve problem (\ref{pro:main}). The
Lagrange of problem (\ref{pro:main}) is given by (\ref{equ:A-Lagrange}) at the bottom of the page,
where $\lambda$ is the Lagrange dual variable associated with
constraint (\ref{const:ANmain_2}). The associated dual function is
defined as
\setcounter{equation}{32}
\begin{equation}\label{equ: A-dual function}
d(\lambda)=\underset{\mathbf{p}\succeq0,\mathbf{\Sigma}_d\succeq
0}\max L(\mathbf{p},\mathbf{\Sigma}_d,\lambda).\vspace{-0.2cm}
\end{equation}
The bisection method used in Section \ref{sec:power allocation-A}
can also be applied to the following dual problem
\begin{equation}\label{pro: A-dual} \underset{\lambda\geq0}\min
~~{d(\lambda)}.
\end{equation}
However, solving problem (\ref{equ: A-dual function}) is still
challenging since the Lagrangian is not concave in
$(\mathbf{p},\mathbf{\Sigma}_d)$. We use the coordinate descent
method \cite{Bertsekas1999} to handle problem (\ref{equ: A-dual
function}), that tries to maximize the Lagrangian by updating
variable $\mathbf{p}$ and $\mathbf{\Sigma}_d$ in an alternating
fashion, to be presented next.


\subsubsection{Update of $\mathbf{p}$ with fixed $\mathbf{\Sigma}_d$}
With $\mathbf{\Sigma}_d$ fixed, the optimal $\mathbf{p}$ to problem
(\ref{equ: A-dual function}) can be obtained by solving the one-dimensional problems (\ref{pro:EA-dualfuc-pi}) given at the bottom of the page.

\subsubsection{Update of $\mathbf{\Sigma}_d$ with fixed $\mathbf{p}$}
Because the solution $\mathbf{p}^\star$ given by
(\ref{pro:EA-dualfuc-pi}) would yield nonnegative
$\{\mathcal{I}(s_i;\!
H_i\!\sqrt{p_i}s_i\!+w_i\!)\!-\!\mathcal{I}(\!s_i;\!\!\frac{G_i\sqrt{p_i}s_i}{\sqrt{\mathbf{b}_i^H\mathbf{\Sigma}_d\mathbf{b}_i\!+1}}
+\hat{v}_i\!), \forall i\}$, with $\mathbf{p}$ fixed, problem
(\ref{equ: A-dual function}) is equivalent to the following
minimization problem:
\setcounter{equation}{35}
\begin{equation}\label{pro: A-Sigma}
\underset{\mathbf{\Sigma}_d\succeq 0}\min~~ \sum\limits_{i=1}^{N}
\mathcal{I}(s_i;\frac{G_i}{\sqrt{\mathbf{b}_i^H\mathbf{\Sigma}_d\mathbf{b}_i+1}}\sqrt{p_i}s_i
+\hat{v}_i)+\lambda\mathrm{tr}(\mathbf{\Sigma}_d).
\end{equation}
For ease of presentation, let us define
\begin{equation}
\!\!\!T_i(\mathbf{\Sigma}_d)\!\triangleq\!
\mathcal{I}(s_i;\frac{G_i}{\sqrt{\mathbf{b}_i^H\mathbf{\Sigma}_d\mathbf{b}_i+1}}\sqrt{p_i}s_i
+\hat{v}_i), ~i=1,2,...,N.
\end{equation}
Next, we apply a successive convex approximation method which
guarantees to yield a stationary point of problem (\ref{pro:
A-Sigma}).

Consider the first-order approximation to $T_i(\mathbf{\Sigma}_d)$.
Let $t_i=(\mathbf{b}_i^H\mathbf{\Sigma}_d\mathbf{b}_i+1)^{-1}$. Then
$T_i(\mathbf{\Sigma}_d)$ becomes a function of $t_i$, i.e.,
$T_i(t_i)$, and its first derivative with respect to $t_i$ is given
by
\begin{equation}
T_i'(t_i)=|G_i|^2p_i\mathrm{mmse}({|G_i|^2p_it_i}).~~~~~~~~(\mathrm{by}~~
\eqref{equ:derivative})
\end{equation}
Then it is readily to see that the first-order approximation of
$T_i(t_i)$ at the point
$\bar{t}_i=(\mathbf{b}_i^H\mathbf{\bar\Sigma}_d\mathbf{b}_i+1)^{-1}$,
where $\mathbf{\bar{\Sigma}}_d$ is the one obtained in the previous
iteration, is given by
\begin{equation}\label{equ:t_appro}
\tilde{T}_i(t_i)=
T_i(\bar{t}_i)+\!|G_i|^2p_i[\mathrm{mmse}({|G_i|^2p_i\bar{t_i}})]
\left(t_i-\! \bar{t_i}\right).
\end{equation}
So we come up with the following first-order approximation to
problem (\ref{pro: A-Sigma}):
\begin{equation}\label{pro:appro}
\begin{array}{ll}
\!\!\!\!\underset{\mathbf{\Sigma}_d\succeq 0, \{t_i\}_{i=1}^N}\min ~&\!\!\!
\sum\limits_{i=1}^{N}\tilde{T}_i(t_i)
+\lambda\mathrm{tr}(\mathbf{\Sigma}_d)\\
~~~~\mathrm{s.t.}&\!\!\!t_i\!=\!(\mathbf{b}_i^H\mathbf{\Sigma}_d\mathbf{b}_i+1)^{-1},~~i=1,2,...,N.
\end{array}
\end{equation}

\begin{table} \caption{} \label{tab:fstord} \centering
\begin{tabular}{l}
\hline SCA Algorithm: Algorithm for solving problem (\ref{pro: A-Sigma}) \\
\hline \textbf{Given:} $\bar{\mathbf{\Sigma}}_d$\\
\textbf{repeat:}\\ ~~~step 1: ~~~~solve problem
(\ref{problem:fstord}) by CVX and obtain optimal
$\mathbf{\Sigma}_d$\\~~~step 2:
~~~~set $\bar{\mathbf{\Sigma}}_d=\mathbf{\Sigma}_d$ \\
\textbf{{until:}}~~a specified convergence criterion is satisfied.\\
\hline
\end{tabular}
\end{table}

Next, we show that the approximated problem (\ref{pro:appro}) can be
reformulated as a convex semi-definite program (SDP). Omitting all
the constant terms, problem (\ref{pro:appro}) can be equivalently
formulated as
\begin{equation}\label{pro:appro-0}
\begin{array}{ll}
\!\!\!\!\underset{\mathbf{\Sigma}_d\succeq 0,\{t_i\}_{i=1}^N}\min ~&\!\!
\sum\limits_{i=1}^{N}|G_i|^2p_i\left[\mathrm{mmse}(\frac{|G_i|^2p_i}{\mathbf{b}_i^H\mathbf{\bar{\Sigma}}_d\mathbf{b}_i+1})\right]
t_i+\lambda\mathrm{tr}(\mathbf{\Sigma}_d)\\
~~~~~~\mathrm{s.t.}&\!\!t_i=(\mathbf{b}_i^H\mathbf{\Sigma}_d\mathbf{b}_i+1)^{-1},~i=1,2,...,N.
\end{array}
\end{equation}
As MMSE is nonnegative \cite{Guo2011}, problem \eqref{pro:appro-0}
is equivalent to the following problem:
\begin{equation}\label{pro: A-Sigma-FO}
\begin{array}{ll}
\underset{\mathbf{\Sigma}_d\succeq 0,\{t_i\}_{i=1}^N}\min ~&
\sum\limits_{i=1}^{N}|G_i|^2p_i\left[\mathrm{mmse}(\frac{|G_i|^2p_i}{\mathbf{b}_i^H\mathbf{\bar{\Sigma}}_d\mathbf{b}_i+1})\right]
t_i+\lambda\mathrm{tr}(\mathbf{\Sigma}_d)\\
~~~~~\mathrm{s.t.}&t_i\geq(\mathbf{b}_i^H\mathbf{\Sigma}_d\mathbf{b}_i+1)^{-1},~~i=1,2,...,N.
\end{array}
\end{equation}
Finally, by applying Shur complement \cite{Boyd2003}, problem
(\ref{pro: A-Sigma-FO}) can be recast as
\begin{equation}\label{problem:fstord}
\begin{array}{ll}
\underset{\mathbf{\Sigma}_d\succeq 0,\{t_i\}_{i=1}^N}\min ~&
\sum\limits_{i=1}^{N}|G_i|^2p_i\left[\mathrm{mmse}(\frac{|G_i|^2p_i}{\mathbf{b}_i^H\mathbf{\bar{\Sigma}}_d\mathbf{b}_i+1})\right]
t_i+\lambda\mathrm{tr}(\mathbf{\Sigma}_d)\\[0.6cm]
~~~~~\mathrm{s.t.}& \begin{bmatrix}
&\!\!\!\!\!\mathbf{b}_i^H\mathbf{\Sigma}_d\mathbf{b}_i+1 &~1\\
&\!\!\!\!1 &~t_i
\end{bmatrix}\succeq \mathbf{0},~~~ i=1,2,...,N.
\end{array}
\end{equation}
Problem (\ref{problem:fstord}) is a standard SDP, which can be
solved by convex solvers such as CVX \cite{Grant2009} or SeDuMi
\cite{Sturm1999}.

We now show the convergence of the proposed successive convex
approximation algorithm. The first-order convex approximation of the
function $T_i(t_i)$, i.e., $\tilde{T}_i(t_i)$ in
\eqref{equ:t_appro}, satisfies the
following three relations: 
\begin{itemize}
\item
$\tilde{T}_i(t_i)\geq T_i(t_i), \forall t_i$, since $T_i(t_i)$ is
concave with respective to $t_i$ \cite{Guo2011};
\item
$\tilde{T}_i(\bar{t_i})= T_i(\bar{t_i})$;
\item
$\tilde{T}_i'(\bar{t_i})= T_i'(\bar{t_i})$.
\end{itemize}
Hence, according to \cite{Marks1978}, the proposed successive convex
approximation algorithm will converge to a point satisfying KKT
conditions of the original problem (\ref{pro: A-Sigma}), thereby
leading to a stationary local maximum of problem (\ref{pro:
A-Sigma}). This successive convex approximation method for solving
problem (\ref{pro: A-Sigma}), called SCA algorithm, is summarized in
Table \ref{tab:fstord}, and the artificial noise aided power
allocation algorithm, called Algorithm 2, for solving
(\ref{pro:main}) is summarized in Table \ref{tab:AE}.

\begin{table}\caption{} \label{tab:AE} \centering
\begin{tabular}{l}
\hline Algorithm 2: Proposed artificial noise aided power allocation scheme \\
~~~~~~~~~~~~~~~~for discrete inputs.\\
\hline \textbf{Given:} $\lambda_h\geq\lambda_l=0$, $\{p_i\}_{i=1}^N$\\
\textbf{{repeat:}}\\ ~~~step 1:~~~~update
$\lambda=\frac{1}{2}(\lambda_l+\lambda_h)$\\~~~step 2: ~~~{repeat:}\\
~~~~~~~~~~~~~~~~~~~obtain the optimal $\mathbf{\Sigma}_d$
using SCA Algorithm in Table \ref{tab:fstord}\\
~~~~~~~~~~~~~~~~~~~and then obtain $\{p_i\}_{i=1}^N$ by solving (\ref{pro:EA-dualfuc-pi})\\
~~~~~~~~~~~~~~~~{until:}~ $L(\mathbf{p},\mathbf{\Sigma}_d,\lambda)$ meets a specified convergence criterion \\
~~~step 3: ~~~~if~
${\frac{1}{N}}\left(\mathrm{tr}(\mathbf{\Sigma}_d)+\sum\limits_{i=1}^{N}p_i\right)<P$,
then update
$\lambda_h=\lambda$, \\
~~~~~~~~~~~~~~~else~update $\lambda_l=\lambda$ \\
\textbf{{until:}} ~~$\lambda$ meets a specified convergence criterion.\\
\hline
\end{tabular}
\end{table}

\subsection{Gaussian Channel Inputs}\label{sec:JPA-AN-Gaus}
The proposed artificial noise aided power allocation algorithm can
be significantly simplified when all $s_i$ are Gaussian. In this
case, the optimal solution $p_i^\star$ in (\ref{pro:EA-dualfuc-pi})
is simply given by \begin{eqnarray}\label{equ:GE_P}
p_i^\star\!=\!\!\begin{cases}\!\! \frac{1}
{2|{H}_i|^2|\hat{G_i}|^2}\!\left(\!\!
\sqrt{\!\hat{C}^2\!-\!
4|{H}_i|^2\!|\hat{G_i}|^2\frac{\lambda\!+|\hat{G_i}|^2\!-\!|{H}_i|^2}{\lambda}}\!-\!\hat{C}\!
\right),\\
~~~~~~~~~~~~~~~~~~~~~~~~~~~~~\textrm{if}~~|H_i|^2\!-\!|\hat{G_i}|^2>\lambda\\[0.2cm]
0,~~~~~~~~~~~~~~~~~~~~~~~~~~~\textrm{otherwise},
\end{cases}
\end{eqnarray}
where
$\hat{G_i}=\frac{G_i}{\sqrt{\mathbf{b}_i^H\mathbf{\Sigma}_d\mathbf{b}_i+1}}$ and $\hat{C}=|{H}_i|^2\!+\!|\hat{G_i}|^2$.
Problem (\ref{pro: A-Sigma}) for this case also reduces to a simple
convex optimization problem as follows:
\begin{equation}\label{pro:eve-noncoh-4}
\mathbf{\Sigma}_d^\star=\arg \underset{\mathbf{\Sigma}_d\succeq
0}\min
\sum\limits_{i=1}^{N}\log_2\left(1+\frac{|G_i|^2p_i}{\mathbf{b}_i^H\mathbf{\Sigma}_d\mathbf{b}_i+1}\right)+\lambda\mathrm{tr}(\mathbf{\Sigma}_d).
\end{equation}
The resulting algorithm, Algorithm 3, for Gaussian inputs is given
in Table \ref{tab:GE}.

\begin{table}\caption{} \label{tab:GE} \centering
\begin{tabular}{l}
\hline Algorithm 3: Proposed artificial noise aided power allocation scheme \\
~~~~~~~~~~~~~~~~for Gaussian inputs.\\
\hline \textbf{Given:} $\lambda_h\geq\lambda_l=0$, $\{p_i\}_{i=1}^N$\\
\textbf{{repeat:}}\\
~~~step 1:~~~  update $\lambda=\frac{1}{2}(\lambda_l+\lambda_h)$\\
~~~step 2:~~~ {repeat:}\\
 ~~~~~~~~~~~~~~~~~~~~~~~~by solving problem (\ref{pro:eve-noncoh-4}) using
CVX\\
 ~~~~~~~~~~~~~~~~~~~~~~~~and then obtain $\{p_i\}_{i=1}^N$ by (\ref{equ:GE_P})\\
~~~~~~~~~~~~~~~{until:} ~$L(\mathbf{p},\mathbf{\Sigma}_d,\lambda)$ meets a specified convergence criterion \\
~~~step 3:~~~ if~
${\frac{1}{N}}\left(\mathrm{tr}(\mathbf{\Sigma}_d)+\sum\limits_{i=1}^{N}p_i\right)<P$,
then update $\lambda_h=\lambda$,\\
~~~~~~~~~~~~~~~else~~~update $\lambda_l=\lambda$\\
\textbf{{until:}}~~$\lambda$ meets a specified convergence
criterion.\\ \hline
\end{tabular}
\end{table}

\subsection{Numerical Results}\label{sec:NRofAN}
\begin{figure}[t]
\centering
\includegraphics[scale=0.45]{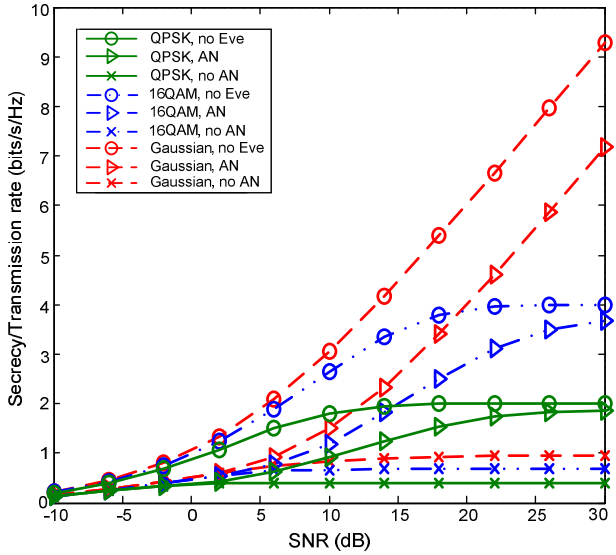}
\caption{Secrecy rates achieved by Algorithm 2 in Table \ref{tab:AE}
using artificial noise (AN) for discrete inputs (denoted by
``{\large$\triangleright$}''), Algorithm 3 in Table \ref{tab:GE}
using artificial noise for Gaussian inputs (also denoted by
``{\large$\triangleright$}''), Algorithm 1 in Table \ref{tab1}
without using artificial noise (denoted by ``$\times$''), and
transmission rate with no eavesdropper achieved by classical/mercury
water-filling strategy (denoted by
``{\large$\circ$}'').}\label{Fig:PA-AN}
\end{figure}
We now show some numerical results to compare the performance of
the proposed artificial noise aided power allocation algorithms
(Algorithm 2 in Table \ref{tab:AE} for discrete inputs and Algorithm
3 in Table \ref{tab:AE} for Gaussian inputs) with the power
allocation algorithm without using artificial noise (Algorithm 1 in
Table \ref{tab1}) and the classical/mercury water-filling strategy
with no eavesdropper (Eve) \cite{Lozano2006}. The OFDM system with
$N = 64$ subcarriers and CP length $N_{cp} = 16$ samples is
considered. The channel coefficients $\{H_i\}$ and $\{G_i\}$ are
generated by following the same procedure in Section
\ref{sec:NRofPA}.
The numerical results are averaged over 300
channel realizations with different values of $\mathrm{SNR}={P}$
(where artificial noise power is included in the total transmit
power $P$).
Figure \ref{Fig:PA-AN} shows the secrecy rate
performance of the three transmission schemes. One can observe from
this figure that Algorithm 2 and Algorithm 3 (denoted by ``
{\Large$\triangleright$}'') using artificial noise outperform
Algorithm 1 (denoted by ``$\times$'') without using artificial noise
for all SNR values, and the amount of performance improvement is
larger for higher SNR, justifying the effectiveness of the
artificial noise design. The performance of the scheme with no
eavesdropper (denoted by ``{\Large$\circ$}'') is expectantly better
than either of Algorithm 2 and Algorithm 3 using artificial noise.
However, the amount of performance difference for both finite
constellations of QPSK and 16QAM is smaller for higher SNR, while it
tends to saturate for higher SNR for the case of Gaussian inputs.
These numerical results have substantiated the efficacy of
Algorithm 2 and Algorithm 3 using artificial noise.

%
%

\section{Conclusions}\label{sec:conclu}
We have presented three efficient transmission schemes for OFDM
wiretap channels. The use of {time-domain} artificial
noise was also considered for jamming
the eavesdropper. These schemes were designed by maximizing secrecy
rate under the constraint of total transmit power, and the Lagrange
dual optimization method was used to resolve the associated nonconvex
optimization problems. The proposed schemes consist of Algorithm 1
without using artificial noise for discrete inputs, Algorithm 2
using artificial noise for discrete inputs, and Algorithm 3 using
artificial noise for Gaussian inputs. The duality gap for Algorithm
1 was shown to decrease with $N$ (the number of subcarriers of OFDM)
in the order $O(1/\sqrt{N}$). {Numerical} results were also provided
to demonstrate that Algorithm 1 significantly outperforms some
existing schemes though artificial noise is not used, and that
Algorithm 2 (for discrete inputs) and Algorithm 3 (for Gaussian
inputs) can further upgrade the secrecy rate performance compared to
Algorithm 1 due to the use of the time-domain artificial noise.

As a future research, it is of interest to consider the
scenario where the CSI of the eavesdropper is unknown or partially
known to the transmitter \cite{Garnaev2009}. Moreover, the
scenario that instead of throwing away the cyclic prefix in standard
OFDM systems, the eavesdropper can employ a more effective receiver,
is worthy of further investigation. In our work, the artificial
noise is only placed in the null space of the legitimate
receiver's channel. The artificial noise design by placing the
artificial noise in a suitable subspace depending on both the
legitimate receiver's and the eavesdropper's channels is {also} left as a
future study.

\section*{Appendix I: Proof of \eqref{key}}
In view of the fact that $\{s_i\}$ are i.i.d., and $\{w_i\}$ and
$\{v_i\}$ are also identically distributed, let us write the secrecy
rate on the frequency $f$ as
\begin{equation}\label{equ:gradients of R}
R_{s}(f,p)\triangleq
[\mathcal{I}(s;H(f)\sqrt{p}s+n)-\mathcal{I}(s;G(f)\sqrt{p}s+n)]^+,
\end{equation}
where $s$ has the same distribution as $s_i$ and $n$ has the same
distribution as $w_i$ or $v_i$. According to Proposition
\ref{proposition:secrecy rate}, we only need to consider the the
frequencies with $H(f)>G(f)$.

{Since $H(f)$ and $G(f)$ are Lipschitz continuous, they are also
uniformly bounded, such that
\begin{equation}\label{equ:bounded_h_g}
|H(f)|\leq M_H,~~|G(f)|\leq M_G, ~~\forall f\in[0,1]
\end{equation}
where $M_H>0$ and $M_G >0$ are constants.
By (\ref{equ:derivative}), (\ref{equ:MMSE}), \eqref{equ:bounded_h_g}
and $\mathrm{mmse}(|H(f)|^2p)\leq E[|s|^2]=1$ \cite{Guo05}, we
have that
\begin{subequations}\label{equ:derivative_R_p}
\begin{eqnarray}
\!\!\!\!\!\!\!\!&&\!\!\!\!\!\!\!\!\frac{\partial{R_{s}(f,p)}}{\partial{p}}\\
\!\!\!\!\!\!\!\!\!\!\!\!\!\!=\!\!\!\!&&\!\!\!\!\!\!\!\!|\!H\!(\!f)|^2\mathrm{mmse}(|\!H\!(\!f)|^2p)\!-\!|G(\!f)|^2\mathrm{mmse}(|G(\!f)|^2p)\\
\!\!\!\!\!\!\!\!\!\!\!\!\!\!\leq \!\!\!\!&&\!\!\!\!\!\!\!\!
|H(f)|^2\leq M_H^2.
\end{eqnarray}
\end{subequations}
According to the Lagrange's mean value theorem \cite{Jeffreys1988},
(\ref{equ:derivative_R_p}) implies that
\begin{equation}\label{equ:R_lipschis_p}
|R_{s}(f,p)-R_{s}(f,p')|\leq M_H^2|p-p'|,
\end{equation}
for any two $p,p'\geq 0$. On the other hand, using \eqref{equ:hi_lipschitz}, \eqref{equ:bounded_h_g}, $\mathrm{mmse}(|H(f)|^2p)\leq1$, and
\begin{equation}
\frac{\partial{\mathcal{I}(s;{H}\sqrt{p}s\!+\!n)}}{\partial{|H|^2p}}
=2p{|H|}\mathrm{mmse}({|H|}^2p),
\end{equation}
we can attain an upper bound of
$|R_{s}(f,p)-R_{s}(f',p)|$, i.e., 
{\small\begin{align}
&~~~~|R_{s}(f,p)-R_{s}(f',p)|\nonumber\\
&\leq
|\mathcal{I}(s;H(f)\sqrt{p}s+n)-\mathcal{I}(s;H(f')\sqrt{p}s+n)|
\notag \\
&~~~~+|\mathcal{I}(s;G(f)\sqrt{p}s+n)-\mathcal{I}(s;G(f')\sqrt{p}s+n)|\nonumber\\
& \leq 2p M_H\big|H(f)-H(f')\big|+2p M_G\big|G(f)-G(f')\big| \nonumber \\
& \leq \left(2pM_HL_H+2pM_G L_G\right){|f-f'|}. \label{equ:R_lipschis_i}
\end{align}}}
By (\ref{equ:R_lipschis_p}) and
(\ref{equ:R_lipschis_i}), we obtain that {\small
\begin{align}
&|R_{s}(f,p)-R_{s}(f',p')|\nonumber\\
=&|R_{s}(f,p)-R_{s}(f,p')+R_{s}(f,p')-R_{s}(f',p')| \notag\\
\leq & |R_{s}(f,p)-R_{s}(f,p')|+|R_{s}(f,p')-R_{s}(f',p')| \notag \\
\leq & M_H^2|p-p'|+(2pM_HL_H+2pM_G L_G)|f-f'| \notag \\
\leq &
\max\{M_H^2,2pM_HL_H+2pM_G L_G\}\bigg(|p-p'|+|f-f'|\bigg),
\end{align}}
which is the same as \eqref{key} in which $$L_R=
\max\{M_H^2,2pM_HL_H+2pM_G L_G\}.$$

\section*{Appendix II: Proof of Proposition \ref{property:frequency AN}}
The objective function of problem (\ref{pro:fre-an}) can be rewritten as follows:
{\small\begin{align}R_{s}^{AN}=&\frac{1}{N}\sum_{i=1}^{N}\big[\mathcal{I}\big(s_i;\frac{H_i\sqrt{p_i}}{\sqrt{|H_i|^2\sigma_{a,i}^2+1}}s_i\!+\!
\tilde{w}_i\big)\nonumber\\
&~~~~~~~~~~-\mathcal{I}\big(s_i;\frac{G_i\sqrt{p_i}}{\sqrt{|G_i|^2\sigma_{a,i}^2+1}}s_i\!+\!\tilde{v}_i\big)\big]^+
\end{align}}where $\tilde{w}_i=(H_i\hat a_i+w_i)/\sqrt{|H_i|^2\sigma_{a,i}^2+1}$
and $\tilde{v}_i=(G_i\hat a_i+v_i)/\sqrt{|G_i|^2\sigma_{a,i}^2+1}$
are Gaussian with zero mean and unity variance. The Lagrangian of
problem (\ref{pro:fre-an}) is given by
{\small$$R_{s}^{AN}+\lambda\left(P-{\frac{1}{N}}\sum\limits_{i=1}^Np_i-{\frac{1}{N}}\sum\limits_{i=1}^N\sigma_{a,i}^2\right)+\sum\limits_{i=1}^N\mu_i
\sigma_{a,i}^2+\sum\limits_{i=1}^N\eta_i p_i~,$$}where
$\lambda\geq0$ and $\{\mu_i\geq0, \eta_i\geq0\}$ are the Lagrange
dual variables associated with the constraints (\ref{equ1-in-proof})
and (\ref{equ2-in-proof}), respectively. Then the optimal solutions
of problem (\ref{pro:fre-an}), denoted by $\{p_i^\star\}$ and
$\{\sigma_a^{2\star}\}$, must satisfy the following KKT necessary conditions
\cite{Bertsekas1999}:
\begin{subequations}
\begin{align}
&\frac{\partial{R_s^{AN}}}{\partial{p_i^\star}}-{\frac{1}{N}}\lambda+\eta_i=0,~~~ i=1,2,...,N,\label{kkt1}\\
&\frac{\partial{R_s^{AN}}}{\partial{\sigma_{a,i}^{2\star}}}-{\frac{1}{N}}\lambda+\mu_i=0,~~~i=1,2,...,N,\label{kkt2}\\
&\lambda\left(P-{\frac{1}{N}}\sum\limits_{i=1}^Np_i^\star-{\frac{1}{N}}\sum\limits_{i=1}^N\sigma_{a,i}^{2\star}\right)=0,\label{kkt5}\\
&\mu_i \sigma_{a,i}^{2\star}=0,~~\eta_i p_i^\star=0,~~~i =1,2,...,N,\label{kkt6}\\
&\lambda\geq 0,~\mu_i\geq 0,~\eta_i\geq 0,~i\!=\! 1,2,...,N,\label{kkt8}\\
&\sigma_{a,i}^{2\star}\geq 0,~{p_i^\star}\geq 0.
\end{align}
\end{subequations}

If $|H_i|\leq |G_i|$ for some $i$, then
$\frac{|H_i|^2{p_i}}{|H_i|^2\sigma_{a,i}^2+1} \leq \frac{|G_i|^2
p_i}{|G_i|^2\sigma_{a,i}^2+1}$, and thus
{\small\begin{align}\label{temp}\!\!\!\!\!\!\left[\!\mathcal{I}\!\!\left(\!\!s_i;\!\!\frac{H_i\sqrt{p_i}}{\sqrt{|H_i|^2\sigma_{a,i}^2+1}}s_i\!+\!
\tilde{w}_i\!\!\right)\!\!-\!\!\mathcal{I}\!\!\left(\!\!s_i;\!\!\frac{G_i\sqrt{p_i}}{\sqrt{|G_i|^2
\sigma_{a,i}^2+1}}s_i\!+\!\tilde{v}_i\!\right)\!\!\right]^+\!\!\!\!\!\!\!=\!0
\end{align}}regardless of the values of $p_i^\star$ and
$\sigma_{a,i}^{2\star}$. Therefore, we can assume that $|H_i|>
|G_i|>0$ for all $i=1,\ldots,N$, without loss of
generality\footnote{As $|G_i|=0$ means no eavesdropper, there is no
need to use artificial noise, i.e., $\sigma_{a,i}^{2\star}=0$.}.

{We use contradiction to prove the result.} Suppose that
\begin{align}\label{eq:3}
\sigma_{a,i}^{2\star}>0
\end{align} for some $i$. Then it is
necessary that $p_i^\star>0$; otherwise we end up with \eqref{temp}
again, and thus $\sigma_{a,i}^{2\star}=0$ is also a feasible
solution. With $\sigma_{a,i}^{2\star}>0$ and $p_i^\star>0$, we have
$\mu_i=0$ and $\eta_i=0$ by (\ref{kkt6}). By (\ref{kkt1}),
(\ref{kkt2}) and (\ref{kkt8}), we further obtain
\begin{align}\label{kkt-derivative}
&\frac{\partial{R_s^{AN}}}{\partial{p_i^\star}}=\frac{\partial{R_s^{AN}}}{\partial{\sigma_{a,i}^{2\star}}}={\frac{1}{N}}\lambda\geq 0.
\end{align}
By \cite[Equation (22)]{Palomar2006} and by using the chain rule,
one can show that
{\begin{align} \!\!\!&\frac{\partial{R_s^{AN}}}{\partial{p_i^\star}}\!=\!{\frac{1}{N}}\left[\frac{|H_i|^2\mathcal{E}_{b,i}(p_i^\star,\sigma_{a,i}^{2\star})}{|H_i|^2\sigma_{a,i}^{2\star}+1}
\!-\!\frac{|G_i|^2\mathcal{E}_{e,i}(p_i^\star,\sigma_{a,i}^{2\star})}{|G_i|^2\sigma_{a,i}^{2\star}+1}\right]\!\geq\!  0,
\label{kkt-deri-sig}\\
\!\!\!&\frac{\partial{R_s^{AN}}}{\partial{\sigma_{a,i}^{2\star}}}\!\!=\!\!{\frac{1}{N}}\!\left[\!\frac{|G_i|^4p_i^\star\mathcal{E}_{e,i}(p_i^\star,\sigma_{a,i}^{2\star})}{(|G_i|^2\sigma_{a,i}^{2\star}+1)^2}\!-\!\frac{|H_i|^4p_i^\star\mathcal{E}_{b,i}(p_i^\star,\!\sigma_{a,i}^{2\star})}{(|H_i|^2\sigma_{a,i}^{2\star}\!+\!1)^2}\!\right]
\!\!\geq\!\! 0,\label{kkt-deri-p}
\end{align}}
where $\mathcal{E}_{b,i}(p_i^\star,\sigma_{a,i}^{2\star})$ and
$\mathcal{E}_{e,i}(p_i^\star,\sigma_{a,i}^{2\star})$, similar to
(\ref{equ:MMSE}), are given by
\begin{equation}\nonumber
\mathcal{E}_{b,i}(p_i^\star,\sigma_{a,i}^{2\star})\!\!=\!\!E\!\!\left[\!\big|s_i\!-\!E\big[s_i|H_i\sqrt{p_i^\star}s_i+H_i\hat{a}_i+w_i\big]\big|^2\right]\!\geq\!
0,
\end{equation}
\begin{equation}\nonumber
\mathcal{E}_{e,i}(p_i^\star,\sigma_{a,i}^{2\star})\!\!=\!\!E\!\left[\big|s_i\!-\!E\big[s_i|G_i\sqrt{p_i^\star}s_i+G_i\hat{a}_i+v_i\big]\big|^2\!\right]\!\geq
0.
\end{equation}

Combining {$\frac{|H_i|^2p_i^\star}{|H_i|^2\sigma_{a,i}^{2\star}+1}>
\frac{|G_i|^2p_i^\star}{|G_i|^2\sigma_{a,i}^{2\star}+1}>0$} (due to
$|H_i|>|G_i|$) and \eqref{kkt-deri-sig} yields
\begin{align}
&\frac{|H_i|^2p_i^\star}{(|H_i|^2\sigma_{a,i}^{2\star}+1)}\frac{|H_i|^2}{(|H_i|^2\sigma_{a,i}^{2\star}+1)}\mathcal{E}_{b,i}(p_i^\star,\sigma_{a,i}^{2\star})\nonumber\\
>&
\frac{|G_i|^2p_i^\star}{(|G_i|^2\sigma_{a,i}^{2\star}+1)}\frac{|G_i|^2}{(|G_i|^2\sigma_{a,i}^{2\star}+1)}\mathcal{E}_{e,i}(p_i^\star,\sigma_{a,i}^{2\star}),
\end{align}
which, however, contradicts with \eqref{kkt-deri-p}. Therefore, \eqref{eq:3} is not true and
$\sigma_{a,i}^{2\star}=0$ for all $i$. By this, the asserted statement is proved.

\section*{Acknowledgment}\nonumber
{The authors would like to thank Fei He, Nicola Laurenti, Wei-Chiang Li, Kun-Yu Wang,
and Shidong Zhou for valuable discussions to improve the quality of this paper.}

\bibliographystyle{IEEEtran}
\linespread{1}\selectfont
\bibliography{Haohao}

\end{document}